\documentclass[lettersize,journal]{IEEEtran}

\usepackage[T1]{fontenc}
\usepackage[cmex10]{amsmath}
\usepackage{color}
\usepackage{tcolorbox}
\usepackage{comment,color,soul}

\usepackage{algorithm}
\usepackage[noend]{algpseudocode}
\usepackage{pgfplots}
\makeatletter
\def\BState{\State\hskip-\ALG@thistlm}
\makeatother
\usepackage{algpseudocode}
\usepackage{cite}
\usepackage[utf8]{inputenc}
\usepackage{algorithm}
\usepackage{graphicx}
\usepackage{caption}
\usepackage{subfigure}
\usepackage{pstricks}
\usepackage{color}
\usepackage{bbm}
\usepackage{tcolorbox}
\usepackage{tikz}
\usepackage{pgfplots}
\usepackage{wrapfig}
\usepackage{lipsum}  % generates 
\usetikzlibrary{positioning}

\usepackage{amssymb}
\usepackage{amsthm}

\usepackage{relsize}
\usepackage{bbm}
\usepackage{bm}
\usepackage[font={small}]{caption}
\usepackage{comment,color,soul}
\usepackage{enumerate} 
\usetikzlibrary{arrows,automata}
\usepackage{amsfonts}

\usepackage{mathtools}
\usepackage{url}

\makeatletter
\def\thm@space@setup{\thm@preskip=2pt
        \thm@postskip=2pt \itshape}
\makeatother

\newtheoremstyle{newstyle}
{} %Aboveskip
{} %Below skip
{\mdseries} %Body font e.g.\mdseries,\bfseries,\scshape,\itshape
{} %Indent
{\bfseries} %Head font e.g.\bfseries,\scshape,\itshape
{.} %Punctuation afer theorem header
{ } %Space after theorem header
{} %Heading

\theoremstyle{newstyle}

\newtheorem*{theorem*}{Theorem}
\newtheorem{lemma}{Lemma}

\theoremstyle{definition}
\newtheorem*{example*}{Example}

\theoremstyle{remark}

\newtheorem*{claim*}{Claim}
\newtheorem{remark}{Remark}

\newcommand{\wtilde}{\widetilde}

\newcommand{\Expc}{\mathbb{E}}
\newcommand{\Prob}{\mathbb{P}}

\newcommand{\norm}[1]{\lVert#1\rVert}

\newcommand{\ssn}{\wtilde{\mathcal{N}}^i}

\makeatletter
\def\widebreve{\mathpalette\wide@breve}
\def\wide@breve#1#2{\sbox\z@{$#1#2$}%
     \mathop{\vbox{\m@th\ialign{##\crcr
\kern0.08em\brevefill#1{0.8\wd\z@}\crcr\noalign{\nointerlineskip}%
                    $\hss#1#2\hss$\crcr}}}\limits}
\def\brevefill#1#2{$\m@th\sbox\tw@{$#1($}%
  \hss\resizebox{#2}{\wd\tw@}{\rotatebox[origin=c]{90}{\upshape(}}\hss$}
\makeatletter

\NewDocumentCommand{\grad}{e{_^}}{%
  \mathop{}\!% \mathop for good spacing before \nabla
  \nabla
  \IfValueT{#1}{_{\mspace{-4mu}#1}}% tuck in the subscript
  \IfValueT{#2}{^{#2}}% possible superscript
}

\DeclareMathOperator*{\argmin}{arg\,min}
\IEEEoverridecommandlockouts

\begin{document}
 \sloppy

               \setlength{\belowcaptionskip}{-6pt}
        \setlength{\abovedisplayskip}{1mm}
        \setlength{\belowdisplayskip}{1mm}
        \setlength{\abovecaptionskip}{1mm}

        \title{Secure and Fault Tolerant Decentralized Learning} 
        
        \author{Saurav Prakash, Hanieh Hashemi, Yongqin Wang, Murali Annavaram, Salman Avestimehr
         \thanks{Saurav Prakash is with the Coordinated Science Laboratory, University of Illinois Urbana-Champaign, Urbana, IL 61801, USA (e-mail: sauravp2@uiuc.edu). Hanieh Hashemi, Yongqin Wang, Murali Annavaram, and Salman Avestimehr are with the Electrical Engineering Department, University of Southern California, Los Angeles, CA 90089, USA (e-mail: hashemis@usc.edu, yongqin@usc.edu, annavara@usc.edu, avestimehr@ee.usc.edu).}
    }
    
%     % The paper headers


\maketitle

\begin{abstract}
Federated learning (FL) has emerged as a promising paradigm for training a global model over the data distributed across multiple data owners without centralizing clients’ raw data. 
However, sharing of its local model update with the FL server during training can also reveal information of a client's local dataset. 
As a result, trusted execution environments (TEEs) within the FL server have been recently deployed in production lines of companies like Meta for \textit{secure aggregation} of local updates. 
However, secure aggregation can suffer from performance degradation due to error-prone local updates sent by clients that become faulty during training due to underlying software and hardware failures. 
Moreover, data heterogeneity across clients makes fault mitigation quite challenging, as even the updates from the normal clients are quite dissimilar. 
Thus, most of the prior fault tolerant methods, which treat any local update differing from the majority of other updates as faulty, perform poorly. 
We propose \textit{DiverseFL} that makes model aggregation secure as well as robust to faults when data is heterogeneous. 
In DiverseFL, any client whose local model update diverges from its \textit{associated guiding update} is tagged as being faulty. 
To implement our novel \textit{per-client criteria} for fault mitigation, DiverseFL creates a TEE-based secure enclave within the FL server, which in addition to performing secure aggregation for carrying out the global model update step, securely receives a \textit{small representative sample} of local data from each client only once before training, and computes guiding updates for each participating client during training. 
Thus, DiverseFL provides security against privacy leakage as well as robustness against faulty clients. 
In experiments, DiverseFL consistently achieves significant improvements (up to $\sim 39\%$) in absolute test accuracy over prior fault mitigation benchmarks. 
DiverseFL also performs closely to \textit{OracleSGD}, where server combines updates only from the normal clients. 
We also analyze the convergence rate of DiverseFL under non-IID data and standard convexity assumptions.
\end{abstract}

\IEEEpeerreviewmaketitle

\section{Introduction}
\label{sec:introduction}
Massive amounts of data is being generated each day by the growing ecosystem of billions of computing devices with sensors connected through the network edge and powered by artificial intelligence (AI). This is shaping the future of both public-interest and curiosity-driven scientific discovery as it has the potential to power a wide range of statistical machine learning (ML) based applications. While ML applications can achieve significant performance gains due to large volumes of user data \cite{devlin2018bert,dosovitskiy2020image}, the training data is distributed across multiple data owners in many scenarios and sharing of user data is limited due to privacy concerns and laws \cite{li2020federated,rieke2020future}.

For enabling decentralized machine learning from user data while preserving data privacy, federated learning (FL) has arose to be a promising approach~\cite{konevcny2015federated,konevcny2016federated,mcmahan2017communication,li2020federated,kairouz2021advances,rieke2020future}. A generic FL algorithm consists of two main steps -- the local SGD updates at each participating client using local data, and, the global model update at the central server using the local updates from the participating clients. These steps are carried in tandem iteratively until convergence. Sharing of raw local updates, however, can leak significant private information of the data at the clients \cite{geiping2020inverting,zhu2019deep}.   

To overcome privacy leakage from local model updates, \cite{bonawitz2017practical} proposed a cryptographic obfuscation algorithm. The principle theme is that in each FL round, each client masks its model update using random pairwise secret masks, such that appropriate cancellations of masks happen at the server when it combines the local updates during aggregation. This protocol, however, is susceptible to clients dropping out during training. As a result, the FL server has to reconstruct secret shares corresponding to the dropped clients and use them to cancel their corresponding masks which are present in the masked updates of the remaining clients, before being able to carry out the final aggregation and model update. This overhead in every FL round, in addition to the overhead in the generation of pairwise masks in the beginning of {each} FL round, slows down the FL training. \iffalse Recently, \cite{so2021turbo} proposed a circular aggregation algorithm for reducing the aforementioned overheads in secure aggregation. However, \cite{so2021turbo} still fundamentally requires pairwise  \fi  

Another critical bottleneck in aggregation of local updates at the server is that some of the clients can send faulty updates during training due to malfunctioning of devices, which can degrade the training performance tremendously \cite{haque2010hard,lamport2019byzantine}. Our focus is on the scenarios where the clients are \textit{honest-but-faulty}, i.e., the clients follow the training protocol honestly but can malfunction during training due to which the computed local update becomes faulty. The motivation for modeling FL clients as honest-but-faulty comes from the many use cases in the real-world where a set of honest clients wish to participate in privacy-preserving joint training to get a reliable ML model with good performance. For example, a recent paper \cite{sheller2020federated} in Nature Scientific Reports considers the critical scenario where a group of accredited hospitals want to develop an ML model, via FL, for cancer research by training over private medical records maintained at the hospitals. 
In these scenarios, hospitals would honestly follow training protocols due to signed agreements and good intentions, and any malicious behavior can seriously affect reputation in public. 
However, faults such as hardware/software errors can arise during training, as studied extensively in prior works such as~\cite{chen2017distributed,blanchard2017machine,yin2018byzantine,su2018securing,guerraoui2018hidden,xie2019zeno,pillutla2019robust,ozdayi2020defending, jin2020stochastic,pan2020justinian, el2020distributed,fung2018mitigating,li2019rsa,he2020byzantine,peng2020byzantine,cao2020fltrust}. For example, hospitals may have limited and unreliable computing capabilities, particularly in remote rural areas, making them quite likely to suffer from faults in training. 

A number of works focused exclusively on addressing faulty behaviors have been proposed for the IID data setting~\cite{chen2017distributed,blanchard2017machine,yin2018byzantine,su2018securing,guerraoui2018hidden,xie2019zeno,pillutla2019robust,ozdayi2020defending, jin2020stochastic,pan2020justinian, el2020distributed}. When data is IID, the model updates received from the normal clients tend to be distributed around the true model update. Hence, detecting faulty updates has been done through distance based schemes or other robust statistical estimation techniques. Data in typical FL settings, however, is non-IID across clients~\cite{zhao2018federated}, and thus achieving fault resiliency becomes even more challenging as the updates from even the normal clients are quite dissimilar. Furthermore, these methods are not compatible with the pairwise masking based secure aggregation in~\cite{bonawitz2017practical}, as they require the FL server to access unmasked client updates. 

Some recent works such as~\cite{fung2018mitigating,li2019rsa,he2020byzantine,peng2020byzantine,cao2020fltrust} deal with heterogeneous data distribution in FL. In \cite{he2020byzantine}, for example, a resampling approach is proposed to improve the performance of existing faulty robust schemes for the IID setting (e.g., Median \cite{yin2018byzantine}) by reducing the inner/outer variations coming from heterogeneity in stochastic updates within/across clients. In the parallel work of \cite{cao2020fltrust}, a trust bootstrapping method is proposed, in which the server collects a small clean training dataset, independently of client data before training. During training, it computes a model update on root dataset and uses it to assign similarity scores to client updates and carry out their weighted average. 

While above schemes typically perform better than their IID counterparts when data is non-IID, we demonstrate in our experiments in Section \ref{sec:experiments} that their performance is limited in comparison to \textit{OracleSGD}, where the omniscient FL server only aggregates the model updates from the clients that are not faulty. Furthermore, similar to their IID counterparts, these papers do not provide privacy protection of the local updates as the FL server has access to unmasked client updates. Prior works of \cite{so2020byzantine,elkordy2022heterosag} propose to address both secure aggregation and fault mitigation in FL. However, the proposed algorithms incur high secure aggregation complexity of $O(N^3)$ ($N$ being the total number of clients) due to pairwise masking (similar to the issue in \cite{bonawitz2017practical} described above), and are also unsuitable for non-IID data as they rely on similarity of non-faulty updates.

The key challenge in making the secure aggregation algorithm of~\cite{bonawitz2017practical} resilient to faulty client updates is that masked client updates are interdependent due to the pairwise additive masking and only the final aggregated update is accessible in the unmasked form. However, prior proposed algorithms for fault mitigation require access to the actual client updates. 

To overcome the dependency of a client's update on other participating clients for carrying out secure aggregation and speed up the secure aggregation process, Trusted Execution Environments (TEEs) \cite{costan2016intel,ARMREALM,alves2004trustzone} within the FL server have been recently proposed and deployed in production lines of companies including Meta~\cite{nguyen2022federated,mo2021ppfl,teesecureagg}. TEEs such as Intel SGX~\cite{costan2016intel} provide \textit{secure isolated execution environments} where data confidentiality and computational integrity of the individual client's application is guaranteed by the hardware and software encryption. Therefore, individual updates can be securely aggregated within the secure TEE enclave at the FL server without requiring pairwise additive masks as in \cite{bonawitz2017practical}, while simultaneously protecting the privacy of client datasets and model updates \cite{nguyen2022federated} from an honest-but-curious FL server. \iffalse While this also enables the implementation of prior fault mitigation strategies in association with secure aggregation, no prior work has evaluated the performance of prior fault mitigation and secure aggregation within the TEEs. \fi 

{\textbf{Our contributions}}: We propose \textit{DiverseFL}, a novel TEE-based solution involving a \textit{per-client} criteria to overcome the aforementioned challenges in \textit{providing fault resiliency to secure aggregation in FL with non-IID data while maintaining close to OracleSGD performance}. In the following, we highlight the main aspects of our proposal.
 
%\begin{enumerate}
%\item 
\begin{itemize}
\item \textit{Per-client fault mitigation}: Rather than rely on the similarity of updates across clients, our work takes the view that similarity between the expected and received updates from a client are better markers for detecting faulty behaviors. During the offline phase before training starts, the FL server asks each client to share with the secure TEE enclave a small, representative sample of its local data having the same proportion of labels as the client data. During training, for each client, a \textit{guiding} update over its TEE sample is computed within the secure enclave, and this is leveraged to estimate whether the corresponding client is faulty. Our main intuition is that the guiding update associated with a client is similar to the model update received from the client if it is normal, while an arbitrary update from a faulty client is quite different from its associated guiding update. Followed by the fault mitigation, the updates from the non-faulty clients are aggregated within the TEE enclave, thus ensuring fault resilient secure aggregation.

\iffalse\quad Based on these ideas, we propose our novel approach for fault mitigation during training that works on a per-client basis. In particular, a client is flagged by the server as faulty if either of our proposed similarity conditions is violated -- (i) the dot product between the client update and its associated guiding update is greater than a pre-defined positive constant, (ii) the ratio of the Euclidean norms of the two is within a pre-defined range. Followed by the fault mitigation, the updates from the non-flagged clients are aggregated within the TEE enclave, thus ensuring fault resilient secure aggregation.
\fi    
\begin{remark}
\label{rmk:1}
DiverseFL's fault mitigation leverages clean representative subsets of clients' local datasets shared securely by clients only once during the offline phase. As discussed previously, we focus on the practical scenarios where the clients are {honest-but-faulty}, i.e., the clients follow the protocols honestly but can have anomalies due to hardware/software faults during training. Hence, the clients provide clean data to the TEE enclave at the FL server. For example, the National Institutes of Health (NIH) can orchestrate FL for oncology research in which the participating accredited hospitals \cite{sheller2020federated} ensure they provide small clean subsets of their medical datasets. 
\end{remark}    
    
%\item 

\item \textit{Secure Enclave for Privacy Protection}: As mentioned previously, DiverseFL enables the clients to privately share their samples with a Trusted Execution Environment (TEE) based secure enclave on the FL server. The data as well as the applications inside the TEE are protected via software and hardware cryptographic mechanisms~\cite{costan2016intel,mckeen2013innovative,alves2004trustzone}. Because of their hardware guaranteed privacy/integrity features, TEEs have been widely used in recent years for privacy-preserving ML on centralized training with private data \cite{tramer2018slalom,hunt2018chiron,kockan2020sketching,mo2021ppfl,dokmai2021privacy}, for secure aggregation in FL (with no fault mitigation) \cite{nguyen2022federated}, and more recently, for sharing raw data among clients for performance improvement in decentralized serverless training \cite{dhasade2022tee} (without any fault mitigation). DiverseFL is the first algorithm that uses a TEE-enclave (based on Intel SGX~\cite{costan2016intel} in particular) on the FL server for providing secure aggregation as well as fault mitigation in FL setting with non-IID data across clients, where it is hard to know whether dissimilarity between client updates is due to faults or due to data heterogeneity. Particularly, the secure enclave  enables the federated clients to verify that their data samples are not leaked to any external party, even to the FL server administrator. Furthermore, computations of the guiding updates, implementation of the per-client criteria of DiverseFL to filter out the faulty updates, aggregation of the non-faulty updates and finally the model update, all steps are securely carried out within the secure enclave. Thus, DiverseFL provides secure aggregation and as we highlight next, it outperforms prior (non-secure) baselines for fault mitigation in non-IID FL by significant margins.
%\item

\item \textit{Experimental Results}: We evaluate neural network training performance with different benchmark datasets (see Section \ref{sec:experiments} for details). Our results exhibit the important aspects of DiverseFL as described next. First, we demonstrate that even when each client shares just $1{-}3\%$ of its local data with the secure TEE enclave, that sample size is sufficient for reliably computing guiding updates and for significantly improving model performance. Next, we experimentally demonstrate the scalability of implementing DiverseFL. For this, we measure the performance of the Intel SGX based TEE for computing guiding updates, and for comparing the performances of TEE and edge-devices to provide the number of clients each TEE can support without any slowdown. We show that a single TEE can support up to $316$ clients. Such analysis is quite useful in deciding number of TEEs that the server needs to use for setting up the secure enclave. We  also provide an additional discussion on further scaling the number of supported clients based on the recent works and improvements related to TEE capabilities. \iffalse This can be useful in   analyze scalability of TEE-assisted DiverseFL implementation.  provide the number of edge-devices each TEE can support without any slowdown in the FL implementation. \fi
%\item 

\item \textit{Convergence Analysis}: We provide a convergence analysis of DiverseFL for non-IID data distribution across clients in the presence of an \textit{arbitrary} number of faulty clients, under standard assumptions such as strong convexity of the local loss functions. For this, we first obtain the probabilistic bounds pertaining to the error between the current model and the optimal global model with respect to each client that satisfies the DiverseFL's per-client similarity criteria in a given round, and then use the intermediate results to prove convergence. The analysis has been provided in Appendix.
\end{itemize}
\section{Problem Setup}
\label{sec:background}
\label{sec:probSetup}
{
In this section, we first describe stochastic gradient descent (SGD) in the context of federated learning in the absence of faulty clients, then present the fault model for the participating clients, and finally present a concise background on the prior applications of Trusted Execution Environments (TEEs) for privacy-preserving machine learning in centralized as well as distributed settings.}
%  \vspace{-3mm}
\subsection{Federated Learning with SGD}
\label{sec:fedsgd}
We consider a federated learning (FL) setup with $N$ client nodes (workers) that are connected to a central FL server (master). For $j{\in}[N]{=}{\{1,\ldots,N\}}$, let ${D}_j$ denote the client $j$'s labelled dataset with $|D_j|=n$ being the number of points in $D_j$. Furthermore, each data point in $D_j$ is drawn from an unknown local data distribution $\mathcal{D}_j$. Since data is typically non-IID across clients in FL \cite{zhao2018federated}, $\mathcal{D}_i$ can be different from $\mathcal{D}_j$ for $i,j{\in} [N]$ and $i{\neq} j$. The primary aim in FL is to solve the following global optimization problem: 
\begin{align}
\label{eq:main_opt}
\theta^{*}&{=}\argmin_{\theta{\in} \Theta} \frac{1}{N}\sum_{j=1}^N F_j(\theta),
\end{align}
where, $F_j(\theta)=\Expc_{\zeta_j{\sim \mathcal{D}_j}}(l(\theta;\zeta_j))$ denotes the local expected loss at client $j{\in}[N]$, and $\Theta{\subseteq}\mathbb{R}^{d}$ denotes the model parameter space. Here, $l(\theta;\zeta_j){\in}\mathbb{R}$ denotes the predictive loss function (such as the cross-entropy loss) for model parameter $\theta{\in} \Theta$ on the data point $\zeta_j$. 

The solution to \eqref{eq:main_opt} is obtained by using an iterative training procedure that involves two key steps. For a general communication round $i{\in}\{1,{\ldots}, R\}$, the FL server selects a subset $\mathcal{S}^{(i)}$ of the clients, where $|\mathcal{S}^{(i)}|{=}C{\leq }N$. Each client $j\in\mathcal{S}^{(i)}$ receives the current model $\theta^{(i-1)}$ from the master, and locally updates it using stochastic gradient descent (SGD) with a learning rate of $\alpha^{(i)}$ for $E$ iterations. Specifically, for $\tau{\in}\{1,\ldots, E\}$, client $j$ carries out a stochastic gradient update $\theta_j^{(i,\tau)} {=} \theta_j^{(i,\tau-1)}{-} \alpha^{(i)} g_j^{(i,\tau)}$, where $\theta_j^{(i,0)}{=}\theta^{(i-1)}$, and $g_j^{(i,\tau)}{=}\grad_{\theta}l(\theta_j^{(i,\tau-1)};\mathcal{M}_j^{(i,\tau)})$. Here,  $\mathcal{M}_j^{(i,\tau)}$ is a mini-batch of size $m$ sampled uniformly at random from ${D}_j$, and $l(\theta_j^{(i,\tau-1)};\mathcal{M}_j^{(i,\tau)})$ denotes the empirical loss over $\mathcal{M}_j^{(i,\tau)}$. In the second step, the master receives the  model update $\Delta_j^{(i)}{=}\theta^{(i-1)}{-}\theta_j^{(i,E)}$ from each client $j{\in}\mathcal{S}^{(i)}$ and carries out the global model update as follows: $\theta^{(i)} = \theta^{(i-1)} - \Delta^{(i)}$, where $\Delta^{(i)}{=}\frac{1}{|\mathcal{S}^{(i)}|}\sum_{j\in\mathcal{S}^{(i)}} \Delta_j^{(i)}$. This combination of local and global training steps is repeated for $R$ communication rounds, where $R$ is a hyperparameter. 

While the participating clients are honest, they can exhibit faulty behaviors during training as described next.
\subsection{Fault Model}
\label{sec:byzantModel}
We assume that clients are \textit{honest-but-faulty}. Particularly, clients follow the protocols honestly and as highlighted in Remark \ref{rmk:1}, clients have clean datasets during the offline phase, i.e., before training begins. However, during training, clients can malfunction or may experience hardware/software errors, which are hard to detect due to the decentralized, distributed and multi-round implementation of FL. Formally stating, for round $i{\in}R$, let $\mathcal{F}^{(i)}{\subset}\mathcal{S}^{(i)}$ denote the set of faulty clients and $\mathcal{N}^{(i)}{=}\mathcal{S}^{(i)}{\setminus}\mathcal{F}^{(i)}$ denote the set of normal clients. Then, $z_j^{(i)}{=}\Delta_j^{(i)}$ for $j{\in}\mathcal{N}^{(i)}$, while for $j{\in}\mathcal{F}^{(i)}$, $z_j^{(i)}{=}*$. Here, $*$ denotes that $z_j^{(i)}$ can be an arbitrary vector in $\mathbb{R}^d$. Our proposal for fault mitigation, while securely aggregating the local updates from the participating clients, leverages the hardware guaranteed privacy/integrity features of Trusted Execution Environments (TEEs), particularly Intel SGX. In this work, we do not consider side-channel attacks on TEEs that are based on probing physical signals such as electromagnetic leaks and power consumption. Furthermore, we assume all the communication channels are secure (e.g. using Transport Layer Security (TLS)~\cite{dierks2008transport,turner2014transport}).\iffalse Furthermore, we note that the set of faulty clients (and likewise the set of normal clients) can vary during training, for example, a client exhibiting faulty behaviors due to some underlying hardware or software failures may recover after a few iterations. Hence, for $i{\in}\{1,\ldots,R{-}1\}$, $\mathcal{F}^{(i)}$ and $\mathcal{F}^{i+1}$ are not necessarily equal. \fi

In the following, we provide a background on TEEs and their applications in privacy-preserving machine learning.  
\subsection{Trusted Execution Environments}
{
Trusted Execution Environments (TEEs) such as Intel SGX \cite{costan2016intel}, ARM REALM \cite{ARMREALM} and ARMTrustZone \cite{alves2004trustzone} provide \textit{secure isolated execution environments}, where data confidentiality and computational integrity of the client's application is guaranteed by cryptographic hardware and software encryption. This has enabled system-based privacy guarantees for artificial intelligence (AI) applications involving sensitive user data. For instance, authors in~\cite{dokmai2021privacy} introduced a genotype imputation tool based on Intel SGX. Recently, TEEs have also been used for privacy-preserving aggregation (albeit without mitigation of faulty clients during training) of local updates from clients in federated learning in production lines of different companies including Meta~\cite{teesecureagg,mo2021ppfl,nguyen2022federated}.

In our proposal, we use Intel SGX that provides important features for privacy and integrity including \textit{remote attestation}, \textit{local attestation}, and \textit{sealing}. Using remote attestation, an encrypted communication channel can be established between the server and the relying party which provides private direct communication between the two. Local attestation is for secure communication between multiple enclaves on the same platform. Sealing provides a secure data saving for transferring data to the untrusted hardware while protecting data privacy.

While Intel SGX has been widely used for privacy-preserving AI applications, computations performed within SGX are restricted to execution on CPU in the current implementations. Moreover, TEEs generally provide a limited amount of secure memory (enclave) that is tamper-proof even from a root client. Our proposal DiverseFL leverages Intel SGX for detecting faulty clients and performing secure aggregation at each communication round. As we describe in the next section, DiverseFL requires low computation load of SGX with respect to each client, thus allowing multiple participating clients to be supported by a single TEE without causing overheads in the wall-clock time.  %For instance, Intel SGX provides $128$ MB as the enclave memory. If the size of the private data exceeds the TEE limit, it will pay a significant performance penalty for encryption and eviction of pages for swapping.
}

%H: removed the following statement as it tends to make federated learning weak: why don't all the client data to the server?
\iffalse Recently cloud providers such as IBM and Microsoft Azure support cloud instances with Intel SGX~\cite{intel}. SGX has various applications in DNN private computations. For instance, these works utilize TEEs for privacy-preserving DNNs inference ~\cite{hanzlik2018mlcapsule, ohrimenko2016oblivious, gu2018securing,tramer2018slalom,mo2020darknetz}, and training~\cite{hashemi2020darknight,hynes2018efficient,hunt2018chiron,ghareh2020mitigating, mo2021ppfl, fereidoonisafelearn}.\fi 

\section{The Proposed DiverseFL Algorithm}
\label{sec:diversefl}
In this section, we describe the various aspects of our proposal DiverseFL in detail. We also illustrate the effectiveness of the core per-client criteria of DiverseFL for fault mitigation. In Fig. \ref{fig:system}, we provide an overview of DiverseFL. \iffalse while an algorithmic description is provided in Algorithm \ref{algo:diverseFL}. \fi
\subsection{Description of DiverseFL}
DiverseFL has five fundamental steps as described next. 

\begin{figure*}[h!]
    \centering
    \includegraphics[scale =1.0,width=\linewidth]{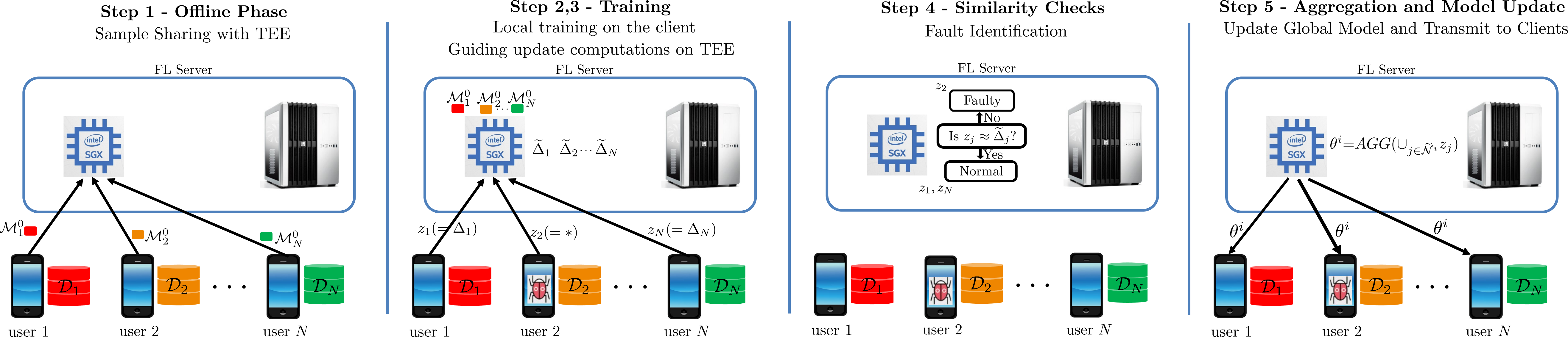}
    %\vspace{-12pt}
    \caption{Illustration of system components and general steps in DiverseFL for communication round $i\in[R]$. Without loss of generality, we have assumed that all clients participate in the communication round $i$. For brevity, we use $AGG(\cdot)$ in final step to jointly denote the aggregation of the potential normal clients as well as the global model update step.}
    \label{fig:system}
    %\vspace{-6pt}
\end{figure*}
{

{\textbf{ {Step 1: Sharing of small samples with TEE:}}} In the offline phase before training, the FL server asks each client $j{\in}[N]$ to draw a mini-batch $\mathcal{M}_j^{(0)}$ of size $s$ from its local dataset, and share it with a TEE-based secure enclave on the FL server, using a mutually agreed encryption key between the TEE and the client. For obtaining a sample mini-batch that is representative of its local dataset, a client first finds the number of features corresponding to each label in the local dataset (as data is non-IID across clients, certain labels may not be available in the client's dataset). Thereafter, the client computes the number of data points that need to be sampled from each local label set so that the final mini-batch of size $s$ has the same proportion of labels as in the local dataset. Finally, from each label set, the client samples the proportionate number of data points uniformly at random, and the final mini-batch is shared securely with TEE. 
\begin{remark}
The TEE enclave code can be verified by each client to guarantee that none of the samples it shares with the TEE leaves the TEE, and that the TEE only uses it for computing guiding updates (described in Step 3 below) during training. Furthermore, this sample sharing step occurs only once, in the offline phase before training begins. We further emphasize that the FL administrator can access neither the data samples or the guiding updates computed within the TEE, thus protecting data privacy. 
\end{remark}

}

In each communication round $i{\in} R$, the following four training steps are carried out.

\textbf{ {Step 2: Training on the clients:}}
Training in DiverseFL proceeds as follows. During $i$-th round, the server selects a subset $\mathcal{S}^{(i)}$ of the clients, and sends the current model $\theta^{(i-1)}$ to each client $j\in\mathcal{S}^{(i)}$. A normal client $j{\in}\mathcal{N}^{(i)}{\subseteq}\mathcal{S}^{(i)}$ performs $E$ local SGD updates as described in Section \ref{sec:fedsgd} to obtain the locally updated model $\theta_j^{(i,E)}$, and computes the model update $\Delta_j^{(i)}{=}(\theta^{(i-1)}-\theta_j^{(i,E)})$. It then uploads $z_j^{(i)}{=}\Delta_j^{(i)}$ to the master. Local updates from the faulty clients $\mathcal{F}^{(i)}{=}\mathcal{S}^{(i)}{\setminus}\mathcal{N}^{(i)}$ can be arbitrary vectors. 

\textbf{ {Step 3: Guiding model update computations on TEE:}} For each client $j{\in}\mathcal{S}^{(i)}$, the TEE on the FL server updates $\theta^{(i-1)}$ via gradient descent for $E$ iterations using the data sample $\mathcal{M}_j^{(0)}$. More formally, for $\tau{\in}\{1,\ldots,E\}$, the TEE obtains $\wtilde{\theta}_j^{(i,\tau)} {=} \wtilde{\theta}_j^{(i,\tau-1)}{-} \alpha^{(i)} \wtilde{g}_j^{(i,\tau)}$, where $\wtilde{g}_j^{(i,\tau)}{=} \grad_{\theta}l(\wtilde{\theta}_j^{(i,\tau-1)};\mathcal{M}_j^{(0)})$ and $\wtilde{\theta}_j^{(i,0)}{=}{\theta}^{(i-1)}$. The guiding model update is then computed as follows: $\wtilde{\Delta}_j^{(i)}$ ${=}$ $(\theta^{(i-1)}{-}\wtilde{\theta}_j^{(i,E)})$. 
\begin{remark}
The guiding model updates in TEE are computed concurrently with the client SGD computations. In our experimental setup, we demonstrate that each Intel SGX based enclave can perform guiding update computations for tens of clients within the time needed by a client to provide its local update to the TEE (see Fig. \ref{fig:TEE1} in Section \ref{sec:timing_analysis} for details). Hence, Step 3 occurs concurrently with Step 2 without any loss in training speed. 
\end{remark}
\textbf{ {Step 4: Fault identification by FL server:}} The secure TEE based enclave within the FL server receives the model updates from the clients and estimates whether the update from client $j{\in} [N]$ is faulty or not based on the extent of \textit{similarity between $z_j^{(i)}$ and $\wtilde{\Delta}_j^{(i)}$}, using $\wtilde{\Delta}_j^{(i)}$ as a surrogate for $ {\Delta}_j^{(i)}$. It considers the following two similarity metrics:
\begin{align}
&\text{Direction Similarity}: C_1=\text{sign}(\wtilde{\Delta}_j^{(i)}\cdot z_j^{(i)}), \label{eq:simi1}\\
&\text{Length Similarity}: C_2=\frac{\norm{z_j^{(i)}}_2}{\norm{\wtilde{\Delta}_j^{(i)}}_2} \label{eq:simi2}.
\end{align}The main idea is that $\wtilde{\Delta}_j^{(i)}$ approximates the model update $ {\Delta}_j^{(i)}$ of a normal client. Therefore, the similarity between $\wtilde{\Delta}_j^{(i)}$ and $z_j^{(i)}{=} {\Delta}_j^{(i)}$ when $j{\in}\mathcal{N}^{(i)}$   (and likewise the dissimilarity between $\wtilde{\Delta}_j^{(i)}$ and $z_j^{(i)}$ when $j{\in}\mathcal{F}^{(i)}$) can be leveraged for fault mitigation. Thus, when $C_1{>}0$, it suggests that $z_j^{(i)}$ is \textit{approximately} in a similar direction as $ {\Delta}_j^{(i)}$ for client $j$. Similarly, when $C_2{\sim} 1$, it suggests that the norms $\norm{z_j^{(i)}}_2$ and $\norm{ {\Delta}_j^{(i)}}_2$ are \textit{approximately} equal, while very large or very small values of $C_2$ suggest a large deviation between $z_j^{(i)}$ and $ {\Delta}_j^{(i)}$. These arguments motivate our following two key conditions for fault mitigation that are verified within the TEE for the local update from each participating client $j{\in}\mathcal{S}^{(i)}$:
\begin{align}
\text{Condition } 1&: C_1{>}\epsilon_1, \label{eq:cond1}\\
\text{Condition } 2&: \epsilon_2{<}C_2{<}\epsilon_3,\label{eq:cond2}
\end{align}
where $\epsilon_1$, $\epsilon_2$ and $\epsilon_3$ are hyperparameters in DiverseFL. Any client $j{\in}\mathcal{S}^{(i)}$ is flagged as a faulty node if either of the above two conditions is not satisfied.

\textbf{ {Step 5: Secure aggregation and global update:}}  Let $\wtilde{\mathcal{F}}^{(i)}$ denote the set of faulty nodes as estimated in Step 4, and likewise let $\wtilde{\mathcal{N}}^{(i)}{=}\mathcal{S}^{(i)}{\setminus}\wtilde{\mathcal{F}}^{(i)}$ be the estimated set of normal clients. The following global model update is executed within the TEE:
\begin{equation}
\label{eq:modUpdateScheme}
    \theta^{(i)} = \theta^{(i-1)} -\frac{1}{|\wtilde{\mathcal{N}}^{(i)}|}\sum_{j{\in}\wtilde{\mathcal{N}}^{(i)}}z_j^{(i)}\iffalse+\lambda \theta^{(i-1)}\fi.
\end{equation}
As the aggregation step is also carried out securely within the TEE at the FL server, DiverseFL ensures that there is no privacy leakage from the local model updates shared by the clients.

 Next, we illustrate the effectiveness of the core per-client criteria of DiverseFL for fault mitigation.
%% #SP the following algorithm will be moved to appendix %%

\subsection{Effectiveness of per-client fault mitigation} 
\label{sec:criteria_effectiveness}
To illustrate the effectiveness of using the two similarity metrics in the detection of a faulty node, we consider the product $C_1{\times}C_2$ of the metrics in equations~\eqref{eq:simi1} and~\eqref{eq:simi2}, and plot its variation for different clients across iterations. For this, we consider the setting of $23$ clients and neural network training with MNIST dataset, in the presence of label flip fault, as described in Section \ref{sec:non_targeted}.  Data is distributed non-IID, a data sharing of $1\%$ is used, and $5$ out of the $23$ clients are assumed to be faulty during training. For examining how the two similarity metrics behave in each round, we consider the oracle algorithm, named OracleSGD, in which only the normal client updates are aggregated for updating the model in each round, i.e. only the updates of the $18$ normal clients are aggregated for secure aggregation and global model update. Further setup details are deferred to Section \ref{sec:non_targeted}. 
\begin{figure}[h!]
%\vspace{-4mm}
\centering
\includegraphics[scale=1.0,width=0.9\linewidth]{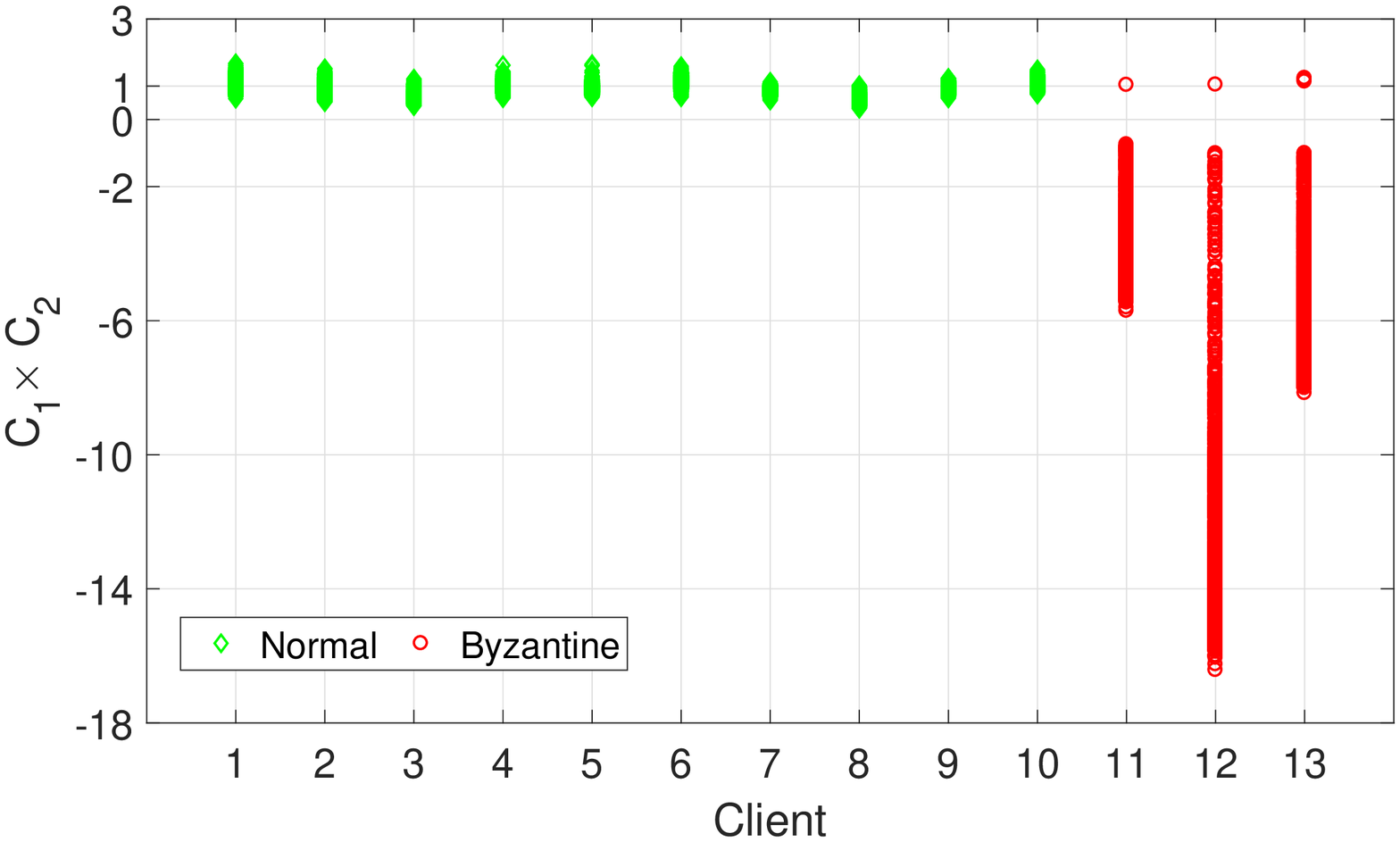}
\caption{The values of $C_1{\times}C_2$ in the $1000$ training rounds are plotted in red for faulty clients, and in green, for normal clients. For normal clients, $C_1{>}0$ exclusively, and $C_2$ is concentrated around $1$. For faulty clients, $C_1{<}0$ in almost all iterations, and $C_2$ varies significantly.}
\label{fig:schemeFig}
% \vspace{-4mm}
\end{figure}

 In Fig. \ref{fig:schemeFig}, we plot the product $C_1{\times}C_2$ for $1000$ training rounds for $10$ different normal clients and $3$ different faulty clients, where green color marker denotes that client is normal, while red color marker denotes that client is faulty. As can be seen from Fig. \ref{fig:schemeFig}, $C_1{>}0$ exclusively and $C_2$ is concentrated around $1$ in all rounds for each normal client. For clients exhibiting faulty behaviour, $C_1{<}0$ \textit{almost} exclusively, and there is a large variation of $C_2$ during training. For example, out of the $1000$ training rounds, $C_1{>}0$ for only $1$ round for both client $11$ and client $12$, and for $3$ rounds for client $13$. Similar results were observed for the other clients. Therefore, condition $1$ presented in \eqref{eq:cond1} is critical for mitigating faults. Condition $2$ presented in \eqref{eq:cond2} is complementary to condition $1$, as it helps to keep the deviations from the true updates low, thus helping to filter the faulty updates in scenarios in which $\Delta_j^{(i)}$ gets simply scaled with a large absolute value due to an underlying software or hardware error. As we demonstrate in experiments in Section \ref{sec:experiments}, the proposed similarity criteria results in superior performance of DiverseFL.

In Appendix, we provide a convergence analysis for DiverseFL to show how the per-client similarity criteria can lead to convergence when data is non-IID. In the following section, we present results from our experiments. 
% \subsection{Convergence Analysis}
% \label{sec:theory}

%%%%%%%%%%%%%%%%%%%%%%%%%%%%%%%%%%%%%%%%%%%%%%%%%%%%%%%%%%

\section{Experiments}
\label{sec:experiments}

\noindent\textbf{Experimental Setup:}  We implement DiverseFL using an SGX-enabled FL server built using Intel (R) Coffee Lake E-2174G 3.80GHz processor. The server has 64 GB RAM and supports Intel Soft Guard Extensions (SGX). \iffalse The server implements the enclave code where client samples are stored and guiding model updates are computed. \fi For SGX implementations, we use Intel Deep Neural Network Library (DNNL) for designing the DNN layers including the Convolution layer, ReLU, MaxPooling, and Eigen library for Dense layer. %We used Keras 2.1.5, Tenseflow 1.8.0, and Python 3.6.8 to implement various DNN models in the FL server. 
%The implementation is available as open source. 
Clients are based on the popular edge device, Raspberry PI 3 consisting of Quad Core Armv7 CPUs, based on BCM2835 hardware, running Debian 10.9 with Linux Kernel 5.10.17. We built the PyTorch for ARMv7 ISA running Linux and installed Torch on Raspberry PI using this build. The link bandwidth between the FL server and each client is 100Mbps. Unless stated otherwise, we consider $N{=}23$ clients, and all clients participate in each round.

\begin{figure*}[ht]
    \centering
    \includegraphics[width=0.98\linewidth]{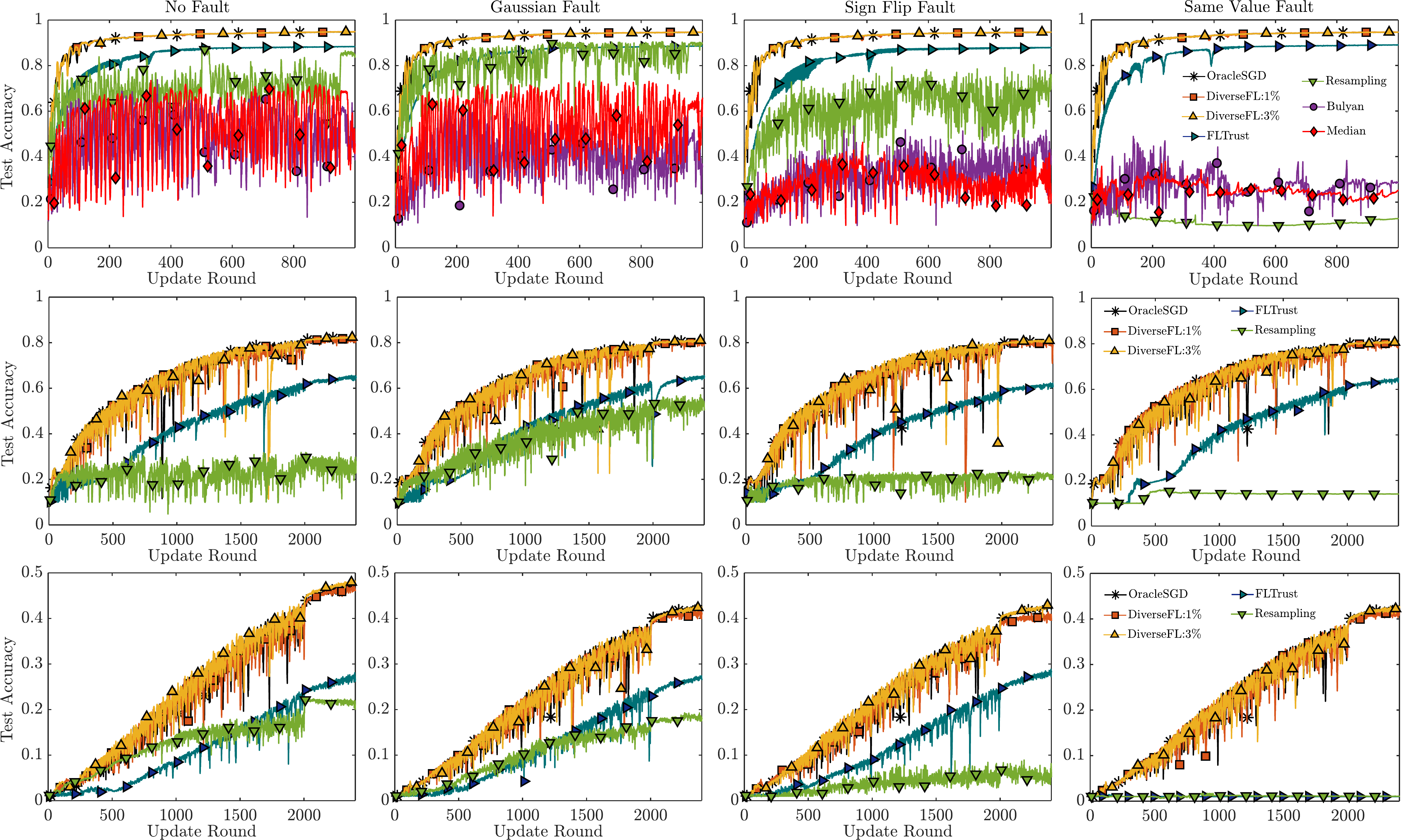}
    %\vspace{-4pt}
    \caption{Top-1 accuracy for neural network training with MNIST (first row of plots), CIFAR10 (second row of plots) and CIFAR100 (third row of plots). DiverseFL with both 1\% and 3\% sample sharing achieves close to OracleSGD performance under all scenarios. Even with relatively simpler training setting with MNIST and a small neural network, prior benchmarks degrade in performance in one or more scenarios. Furthermore, the three sets of results demonstrate that increasing complexity of dataset and training model increases the performance gap of prior benchmarks for non-IID setting.}
    \label{fig:combined_non_targeted}
    %\vspace{-4pt}
\end{figure*}
%, which has performance capabilities similar to mobile devices, with the added advantage that latest Raspberry PI nodes supports running PyTorch code. 
%It has 1GB of DRAM and a 64G Western Digital PiDrive connected through USB. Parameter server consisted of an Intel(R) Coffee Lake E-2174G 3.80GHz processor and supports Intel Soft Guard Extensions (SGX) to evaluate the performance of TEE in computing the guiding gradients on the samples. Intel SGX provides $128 MB$ of secure enclave. Also, the link bandwidth between the central server and each client is 100Mbps. 

The secure TEE enclave in DiverseFL enables secure aggregation by default. Therefore, in our first set of experiments in Section \ref{sec:nn_training}, we demonstrate the superiority of DiverseFL in fault mitigation by comparing it with SOTA algorithms that consider only fault mitigation and not secure aggregation. Thereafter, in Section \ref{sec:timing_analysis}, we present detailed wall-clock timing analysis for implementing DiverseFL in our TEE-assisted federated learning setup and demonstrate the scalability of our proposal in practice. Finally, we provide ablation studies for DiverseFL in Section \ref{app:ablation}.

\subsection{Performance of DiverseFL}
\label{sec:nn_training}
\label{sec:non_targeted}
\noindent{\textbf{{Schemes}}}: In the following, we first summarize the benchmark schemes that we simulate for comparison with DiverseFL.  
\textit{OracleSGD}: The identities of the faulty clients are known at the FL server and the global model is updated using the aggregate of weight updates from only the normal clients. \textit{Median} \cite{yin2018byzantine}: An element-wise median of all the received updates from the clients is computed, which is used to update the global model. \textit{Bulyan} \cite{guerraoui2018hidden}: FL server applies Krum \cite{blanchard2017machine} recursively to remove a set of ${2f}$ possibly faulty clients, and then applies an element-wise trimmed mean operation to filter out ${2f}$ entries along each dimension. \textit{Resampling} \cite{he2020byzantine}: A novel resampling approach of \cite{he2020byzantine} is leveraged to first create a set of $N$ modified updates. For each of them, FL server samples $S_R$ number of client updates uniformly at random and averages them. It then computes the final update by applying Median on the $N$ modified updates. \textit{FLTrust} \cite{cao2020fltrust}: The FL server computes a model update $\wtilde{\Delta}_M$ on its root dataset $\mathcal{D}_M$, a clean small training dataset collected independently of clients' data before training. Client updates are projected onto the root update, and then a weighted aggregation of them is carried out. For best case scenario of FLTrust, we construct root dataset by randomly selecting a subset of the training dataset. 

\noindent{\textbf{{Fault types}}}: For a faulty client $j{\in}[N]$ in round $i{\in}R$, we consider four popular faults.
\textit{Gaussian}: The message to be uploaded gets set to $z_j^i$, where the elements follow a Gaussian distribution with mean $0$ and standard deviation $\sigma_{\text{G}}$. 
\textit{Sign Flip}: The sign of each entry of the local model update is flipped before uploading to the server.
\textit{Same Value}: The message to be uploaded gets set to $z_j^i=\sigma_{\text{S}}\mathbf{1}$, where $\mathbf{1}{\in}\mathbb{R}^d$ is  all-one vector.
\iffalse\textit{Label Flip}: Class label $c$ of each data point in the local training mini-batches gets flipped to  $c_n{-}c$, where $c_n$ is the number of classes during data reading at training time.\fi

\noindent{\textbf{Datasets, Models and Hyperparameters}}: We consider three different benchmark datasets -- MNIST, CIFAR10, and CIFAR100 \cite{mnist2010, cifar}. While MNIST and CIFAR10 have $10$ classes each, CIFAR100 has $100$ classes. For simulating non-IID, for each dataset, the training data is sorted as per class and then partitioned into $N$ subsets, and each subset is assigned to a different client. For MNIST, we use a neural network with three fully connected layers (referred to as 3-NN), and each of the two intermediate layers has $200$ neurons. For CIFAR10/CIFAR100, we use VGG network \cite{simonyan2014very}, specifically VGG-11. We replace the batch normalization layers with group normalization layers as batch norm layers have been found to perform suboptimally when data is non-IID. For group norm, we set each group parameter such that each group has 16 channels throughout the network. Glorot uniform initializer is used for weights in all convolutional layers, while default initialization is used for fully connected layers. 

For Gaussian and same value faults for each dataset, we set $\sigma_{\text{G}}{=}\sigma_{\text{S}}{=}10$. For MNIST, we set $R$ to $1000$, use an initial learning rate of $0.06$ with step decay of $0.5$ at rounds $500$ and $950$. For CIFAR10 and CIFAR100, warmup is used for the first $1000$ rounds, increasing the learning rate linearly from $0.05$ to $0.1$. Also, number of rounds $R{=}2400$ and learning rate is stepped down by a factor of $0.4$ at iteration $2000$. 

For each scenario, local training batch size is $10\%$ of the local dataset, regularization $\lambda=0.0005$, $E=1$. For Resampling, $S_R=2$ and for FLTrust, $1\%$ random subset of training data is used as root dataset. 
For DiverseFL, we consider two sampling size scenarios of $1\%$ and $3\%$ of the local dataset. Furthermore, we keep $(\epsilon_1{,}\epsilon_2{,}\epsilon_3)=(0{,}0.5{,}2)$ in \textit{all our experiments}. We obtained these values through an initial simple experiment, as described in Sec. \ref{sec:criteria_effectiveness}, which demonstrated that these values could be highly efficient for filtering faults. In our extensive evaluations, these thresholds lead to DiverseFL's consistently superior and near OracleSGD accuracy in all our experiments, thus demonstrating the universality of these values.

\noindent
{\textbf{Results}}: Fig. \ref{fig:combined_non_targeted} illustrates the results for MNIST (with 3-NN), CIFAR10 (with VGG-11) and CIFAR100 (and VGG-11), which have 10, 10, 100 classes respectively. For MNIST, DiverseFL (both 1\% and 3\% sampling rates) outperforms prior approaches by significant margins in all cases of faults, and almost matches the performance of OracleSGD. For CIFAR10, DiverseFL with both $1\%$ and $3\%$ sampling sizes perform quite close to OracleSGD, with the $3\%$ case performing slightly better. Note that for CIFAR10 and CIFAR100, we consider mainly the prior baselines for non-IID data, as both Median and Bulyan, which are primarily fault tolerant approaches for the IID data, perform quite poorly across all faults. Starting with the complex CIFAR100 with large number of classes in the dataset, the sampling rate of $3\%$ provides better performance than 1\%. Nevertheless, even with $1\%$ sampling, DiverseFL outperforms all prior schemes by significant margins. 

It is interesting to note that for both Resampling and FLTrust, the margin of convergence performance from OracleSGD is much larger for CIFAR10/CIFAR100 than for MNIST, due to greater dataset complexity and larger neural networks that introduce much larger variations, both within and across client updates. Hence, DiverseFL is ideally suited for complex models that are becoming more common in FL setting. \iffalse It is an interesting future direction to theoretically quantify the relationship between sample size and number of classes. \fi We also note that as FLTrust utilizes the root update to normalize and aggregate the client updates, it achieves stable convergence and improved accuracy in comparison to prior schemes. However, as each client's data distribution is quite different from root data, projection of client updates on the root update leads to loss of information due to which performance of FLTrust is much lower than that of DiverseFL. 
%  \vspace{-3mm}
\subsection{Scalability of DiverseFL}
\label{sec:timing_analysis}
In the previous subsection, we have demonstrated that DiverseFL's accuracy can consistently outperform prior schemes in the presence of faults. DiverseFL, however, does put a computational overhead on the TEE as in addition to leveraging TEE for secure aggregation, DiverseFL also requires computation of the guiding updates for fault mitigation. In this section, we demonstrate that each TEE can accommodate the guiding update computations of many dozens of clients \textit{without causing any additional latency delays}. To quantify this aspect, we used Raspberry Pi 3 as an edge device as commonly used in prior works~\cite{zhao2018deepthings, wang2018edge}. 
\begin{figure}[h]
\centering
\includegraphics[width=1.0\linewidth]{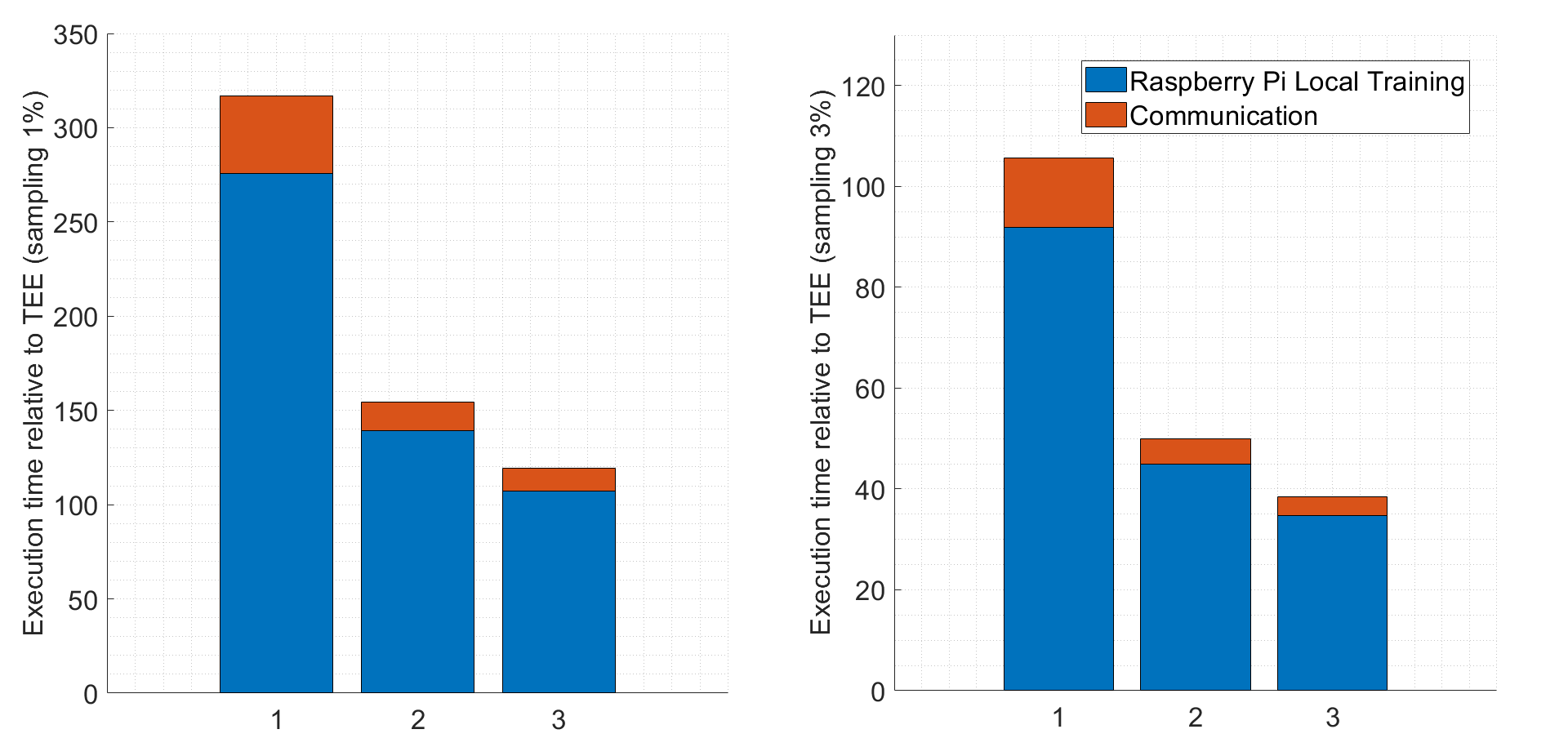}
\caption{Execution time on client (computation + communication time) relative to TEE's guiding update computation. 1: MNIST/3-NN, 2: CIFAR10/VGG-11, and 3: CIFAR100/VGG-11. 1\% sampling used  in (a) and 3\% sampling in (b). Single TEE supports many clients without stalling FL execution.}
\label{fig:TEE1}
%  \vspace{-3mm}
\end{figure}

Fig. \ref{fig:TEE1} illustrates the details of the timing analysis for the different networks and datasets considered in the previous subsection. Fig.~\ref{fig:TEE1}-(a) shows for different models and datasets, the relative execution time of a client device training compared to the SGX execution time to compute guiding update of a single client with a sampling rate of 1\%. Here, the edge device training speed is split into two components: update computation time and the time to communicate the update with the FL server. Consider for example CIFAR10/VGG-11 scenario. Fig.~\ref{fig:TEE1}-(a) shows that the TEE computation is $150$ times faster than an edge device's combined computation and communication time. Thus, the TEE can support $150$ clients. Similarly, the TEE can support about $119$ clients with CIFAR100. We also observe that TEE's relative performance is lower with VGG-11 compared to 3-NN. This is because of the memory limitation of TEE (128 MB in our current implementation). Hence, when the model does not entirely fit within TEE at once due to its size, causing some memory overheads inside TEE.

Fig.~\ref{fig:TEE1}-(b) illustrates that by increasing the sampling rate to $3\%$, the maximum number of clients a TEE can support decreases since the number of data samples that a TEE needs to process increases. Nevertheless, a single TEE can still support between $38$ to $105$ clients based on the model and data size. To scale out the system to support even more clients, one can utilize more instances of the TEEs. Since guiding update computations are done on a per-client basis this scaling approach to multiple TEEs is quite efficient without any substantial synchronization overheads. 

As we have shown above, one TEE can support many clients without causing any performance degradation. Moreover, an FL entity like Meta can easily use more TEEs to support many more clients in each communication round. While not the focus of our work, recent advancements that have greatly improved TEE capabilities can be easily combined with DiverseFL for faster TEE processing for further scalability. For example, 3rd Generation Intel® Xeon® Scalable Processor platforms support 1 TB of secure memory \cite{sgx1tb}, as opposed to the 128 MB TEE used in our FL setup. Also, SGX can support multi-threading that can further improve the performance of TEE operations \cite{tramer2018slalom}. Furthermore, there are recent mechanisms (e.g., DarKnight \cite{hashemi2021darknight}, AsymML \cite{niuasyml}) that securely offload the computation intensive operations to the faster GPUs. Above orthogonal advancements can be combined with DiverseFL for scaling TEE performance at the FL server, in terms of supporting both larger models as well as larger number of clients.  

%As we demonstrated even with this worst case scenario, the system performs well and it is still scalable. 
%  \vspace{-2mm}
\subsection{Ablation Studies}
\label{app:ablation}
We present results from our ablation studies for DiverseFL with respect to different hyperparameters, focusing here on CIFAR10 dataset and Gaussian fault. The trends for other datasets and faults were quite similar to those presented below and have been excluded for brevity.   
\subsubsection{Number of Faulty Clients}
\label{sec:byzant_resiliency}
The per-client fault mitigation approach makes DiverseFL applicable to an arbitrary number of faulty clients. To demonstrate this in practice, we consider the setting in Section \ref{sec:non_targeted} with $N=23$ clients, and report the final test accuracy for both OracleSGD and DiverseFL (with $3\%$ sample sharing), for $f=5$ as well as for $f=17$ (which is equivalent to $\sim 75\%$ faulty nodes in the system). As demonstrated by the results in Tables \ref{tab:cifar10_increase_byzant}, DiverseFL almost matches the performance of OracleSGD even when more than a majority of nodes are faulty.
\begin{table}[h!]
\caption{Final test accuracies for CIFAR10 under Gaussian fault for different number of faults. Similar results were observed for other faults and datasets.}
\label{tab:cifar10_increase_byzant}
\centering
\begin{tabular}{|c|c|c|c|c|}
\hline
& \multicolumn{2}{c|}{Test Accuracy ($\%$): $f{=}5$} & \multicolumn{2}{c|}{Test Accuracy ($\%$): $f{=}17$} \\
\hline
 & OracleSGD & DiverseFL & OracleSGD & DiverseFL \\
\hline
Gaussian & 81.0 & 80.8 & 28.5 & 28.5 \\
\hline
\end{tabular}
\end{table}

\subsubsection{Random Sampling and Multiple Local Iterations}
\label{sec:all_multiple}
In FL, it is a common practice to have multiple local iterations in each communication round, i.e., to have $E>1$ to reduce communication rounds. Furthermore, the FL server typically samples a subset of clients to participate in each communication round. Hence, we evaluate the performance of DiverseFL when multiple local SGD training steps are implemented in each communication round, and assume that the server samples uniformly at random a fraction of the total clients in each communication round. For this, we consider CIFAR10 dataset and the VGG-11 model as described in Section \ref{sec:non_targeted}. Inspired by the simulation setting described in \cite{mcmahan2017communication}, we consider an FL setting with $N=100$ clients and randomly select $25$ of them which are faulty throughout training. Furthermore, to model heterogeneity along the lines of~\cite{mcmahan2017communication}, the training dataset is first sorted as per class, then partitioned into $k N$ shards. Then, each client is assigned $k$ shards randomly without replacement. We consider $E{\in}\{1,4\}$, a sampling size of $3\%$ for DiverseFL, carry out training for a total of $R=5000$ rounds, and the learning rate is stepped down by a factor of $0.4$ at iterations $\{2000,3000,4000\}$. All other hyperparameters are as described in Section \ref{sec:non_targeted}. 

For comparison, we consider the OracleSGD scheme, with $E=4$. The results are illustrated in Fig. \ref{fig:cifar10_nn_multiple}.  
\begin{figure}[h]
    \centering
    \includegraphics[width=0.99\linewidth]{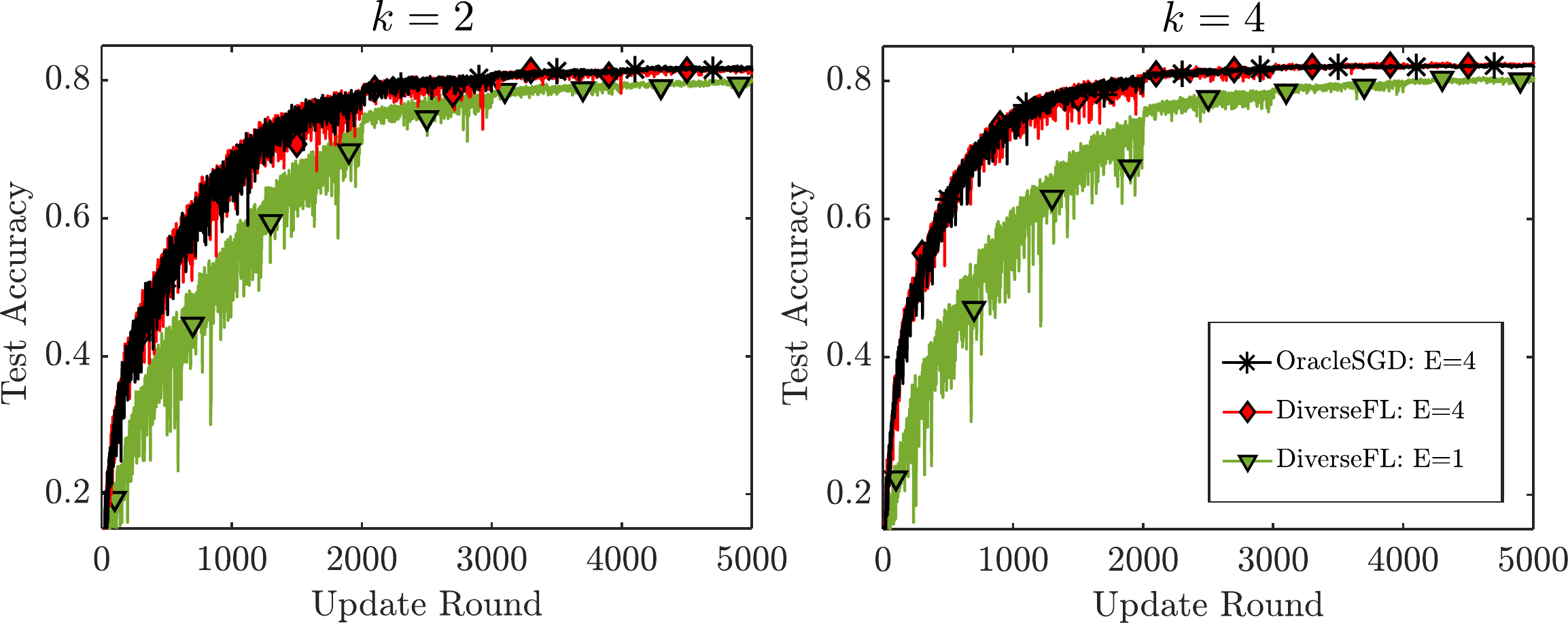}
    %\vspace{-20pt}
    \caption{Performance evaluation of DiverseFL with different number of local iterations with OracleSGD. }
    \label{fig:cifar10_nn_multiple}
\end{figure}
As demonstrated by Fig. \ref{fig:cifar10_nn_multiple}, when $E$ is increased, DiverseFL provides better convergence rate per communication round, further showing that DiverseFL is well suited for federated learning. Additionally, DiverseFL maintains its superior faulty resiliency, as demonstrated by its close to OracleSGD performance when $E=4$ under Gaussian fault. Additionally, DiverseFL is consistent across different data heterogeneity scenarios, demonstrating the wide applicability of DiverseFL for fault resilient secure aggregation in FL.
% \fi

\section{Conclusion}
\label{sec:conclusions}
We consider the problem of making secure aggregation in federated learning (FL), resilient to faults that arise at the clients during the iterative FL training procedure. The problem is particularly quite challenging when clients have data heterogeneity, as even the updates from the normal clients are quite dissimilar. As a result, most of the prior fault tolerant methods, which leverage the similarity among the updates from the benign clients as a marker to determine faults, degrade significantly in performance. We propose a Trusted Execution Environment (TEE) based novel solution, named DiverseFL, that simultaneously ensures fault resiliency and secure aggregation at the FL server, even when data across clients is heterogeneous. As a key contribution, we develop a per-client approach for fault mitigation. It leverages during training, the similarity of the model update received from a client and its associated guiding update, which is computed inside the TEE (at FL server) on a small representative sample of client's local data that the client shares securely with TEE only once during the offline phase before training. Clients whose local model updates diverge from their associated guiding updates are tagged as being faulty, while the remaining are aggregated within the TEE enclave for updating global model. The TEE-based enclave enables client model privacy and protection against potential privacy leakages from sample sharing and guiding update computations, thus enabling secure aggregation as well as fault mitigation. We demonstrate through extensive experimental results that DiverseFL improves the model accuracy and fault resiliency of secure FL with non-IID data significantly as compared to the prior benchmarks. Furthermore, we provide a convergence analysis of DiverseFL under data heterogeneity.
\iffalse
Motivated by the problem of making secure aggregation in federated learning (FL), with non-IID data across clients,  resilient to faults ,   
\fi
\iffalse
Motivated by the problem of making secure aggregation in federated learning (FL) (with non-IID across clients) resilient to faults that arise at the participating clients during training, we propose a Trusted Execution Environment (TEE) based novel solution named DiverseFL. As a key contribution, that the novel per-client approach of fault mitigation in DiverseFL improves the model accuracy and fault resiliency of secure FL with non-IID data significantly as compared to the prior benchmarks, that mainly rely on the similarity of updates across clients.  Furthermore, TEE-based enclave enables client model privacy and protection against privacy leakage. We provide extensive experimental results demonstrating gains of up to $\sim16\%$ in absolute test accuracy, in comparison to prior benchmark schemes. We also provide a convergence analysis of DiverseFL for non-IID data distribution.

\fi

\bibliographystyle{ieeetr}
%vspace{-2mm}
\bibliography{main}
\appendix

%  \vspace{-3mm}
\section{Convergence Analysis}
\label{app:thm}

We now provide a convergence analysis for DiverseFL for non-IID data distribution, adapting the proof developed in \cite{cao2020fltrust}. While \cite{cao2020fltrust} uses trust bootstrapping using a small root dataset, in DiverseFL, a per-client criteria is applied, which makes our analysis different. Furthermore, \cite{cao2020fltrust} assumes the root dataset and clients have samples belonging to the same data distribution, while we consider non-IID data across clients.  We recall from Section \ref{sec:probSetup} that $\mathcal{D}_j$ denotes the  data distribution at client $j\in[N]=\{1,\ldots,N\}$, and the optimal model $\theta^*\in \Theta$ minimizes the loss function $\frac{1}{N}\sum_{j=1}^N F_j(\theta)$, where $F_j(\theta)=\Expc_{\zeta_j}(l(\theta;\zeta_j))$ denotes the local loss function at client $j{\in}[N]$ corresponding to the   local data distribution at client $j$. In the following, we list the standard assumptions of our convergence result:\\
% \begin{comment}
% $$\frac{1}{2}{||}hello{||}$$
% $$\left\|\frac{1}{2}\right\|$$
% $$\norm{\frac{1}{2}}$$
% $$\frac{\mu}{2}\norm{\theta-\hat{\theta}}^2$$
% $$\frac{\mu}{2}\left\|\theta-\hat{\theta}\right\|^2$$
% \end{comment}
\textit{Assumption 1} ($\mu$-Strong convexity). \textit{For any $\theta$, $\hat{\theta}{\in}\Theta$, $j{\in}[N]$:
\begin{equation}
    F_j(\theta)\geq F_j(\hat{\theta})+\left<\grad F_j(\hat{\theta}),\theta-
\hat{\theta}\right>+\frac{\mu}{2}\norm{\theta-\hat{\theta}}^2.
\end{equation}
}

\noindent\textit{Assumption 2} ($L$-Lipshitz continuity). \textit{For any $\theta$, $\hat{\theta}{\in}\Theta$, $j{\in}[N]$:
\begin{equation}
    \norm{\grad F_j(\theta)-\grad F_j(\hat{\theta})}  \leq L\norm{\theta-\hat{\theta}}.
\end{equation}
Furthermore, let $\mathcal{M}_j{\sim}\mathcal{D}_j$ denote any batch of data from the   data distribution $\mathcal{D}_j$ of client $j\in[N]$ such that $|\mathcal{M}_j|\geq s$. Then, similar to the empirical form of Lipshitz continuity used in the convergence analysis in \cite{cao2020fltrust}, for any $\delta>0$, we make the following assumption: 
\begin{equation}
    \Prob\left(\sup_{\theta,\hat{\theta}\in\Theta, \theta\neq\hat{\theta}}\frac{\norm{l(\theta;\mathcal{M}_j)-l(\hat{\theta};\mathcal{M}_j)}}{\norm{\theta-\hat{\theta}}}\leq L_1\right)\geq 1-\frac{\delta}{3}.\nonumber
\end{equation}
}

\noindent\textit{Assumption 3} (Boundedness). \textit{Let $\mathcal{M}_j{\sim}\mathcal{D}_j$ denote any batch of data from $\mathcal{D}_j$ for $j{\in}[N]$, and let $h(\mathcal{M}_j,\theta){=}\grad l(\theta;\mathcal{M}_j)-\grad l(\theta^*;\mathcal{M}_j)$ for $\theta{\in}\Theta$. Let $B{=}\{v{\in}\mathbb{R}^d{:}\, \norm{v}=1\}$ denote the unit sphere. We assume that for any $\theta{\in}\Theta$, $\theta{\neq}\theta^*$ and for any unit vector $v{\in} B$, $\grad l(\theta^*;\mathcal{\mu}_j)\cdot v$ is sub-exponential with $\sigma_1$ and $\gamma_1$, while ${(}h(\mathcal{M}_j,\theta)-\mathbb{E}(h(\mathcal{M}_j,\theta)){)}  \cdot$ $v{/}\norm{\theta-\theta^*}$ is sub-exponential with $\sigma_2$ and $\gamma_2$, where $x_1\cdot x_2$ denotes the dot product between $x_1$ and $x_2$. More formally, there exist positive constants $\sigma_1,\sigma_2,\gamma_1, \text{and}\, \gamma_2$ such that for any $\theta\in\Theta$, $\theta\neq\theta^*$, for any unit vector $v\in B$, and $\forall |\xi|\leq \min\{1/\gamma_1,1/\gamma_2\}$, we have the following:}
\begin{equation}
    \sup_{v\in B}\Expc\left(\exp\left({\xi \left(\grad l(\theta^*;\mathcal{M}_j)\cdot v\right)}\right)\right)\leq \exp\left(\frac{\sigma_1^2 \xi^2}{2}\right),\nonumber
\end{equation}  
\begin{align}
            \sup_{v\in B,\theta\in\Theta}&\Expc\left(\exp\left({\frac{\xi \left(h(\mathcal{M}_j,\theta)-\Expc(h(\mathcal{M}_j,\theta))\cdot v\right)}{\norm{\theta-\theta^*}}}\right)\right)\nonumber\\
            &\leq \exp\left(\frac{\sigma_2^2 \xi^2}{2}\right).
\end{align}
\textit{Furthermore, we assume $\beta$-boundedness on data heterogeneity}: 
\begin{equation}
    \norm{\grad F_j(\theta)-\frac{1}{N}\sum_{j\in[N]}\grad F_j(\theta)}\leq \beta,\,\forall j\in[N],\,\forall\theta \in \Theta.
\end{equation}

\textit{Additionally, we assume that the model parameter space $\Theta$ is bounded, i.e., there exists $r{>}0$ such that $\Theta{\subset}\{\theta{:} \norm{\theta{-}\theta^*}{\leq} r\sqrt{d}\}$.   }
\iffalse Under Assumptions $1-3$, we prove the convergence of DiverseFL, under the following additional assumptions: (1) each normal client executes a single local SGD iteration in each communication round, (2)  all clients participate in each round of training, and (3) each client shares a clean sample with the server.
{Next, we present our convergence result.}\fi 
\begin{theorem*}
\label{thm:convergence}
Assume $\epsilon_1=0$, $\epsilon_2=1/\epsilon_3$, $E=1$, $|\mathcal{S}^{(i)}|=N$ $\forall i\in[R]$, $\mathcal{M}_j^{(0)}\sim \mathcal{D}_j$, $|\mathcal{M}_j^{(0)}|=s$ $\forall j\in[N]$ and Assumptions $1$, $2$ and $3$ hold. Then, for an arbitrary number of malicious clients, for any $\delta>0$ and with a constant learning rate of $\alpha{=}\mu/(2L^2)$, DiverseFL achieves the following error bound with a probability of at least $1{-}\tilde{\delta}$:  
\begin{equation}
\label{eq:thm}
    \norm{\theta^i-\theta^*}\leq (1-\rho)^i \norm{\theta^0-\theta^*}+\frac{\alpha(2+\epsilon_3)(4\Gamma_1+\beta)}{\rho},
\end{equation}
where $\theta^i$ denotes the global model after communication round $i$, \\
$\Gamma_1=\sigma_1\sqrt{\frac{2}{s}}\sqrt{d\log(6)+\log(3/{\delta})}$, ${\delta}=\frac{\tilde{\delta}}{N}$,\\
$\rho=1-\left(\sqrt{1-\frac{\mu^2}{4L^2}}+8\alpha(2+\epsilon_3)\Gamma_2+\alpha(1+\epsilon_3) L\right)$,\\
$\Gamma_2=\sigma_2\sqrt{\frac{2}{s}}\sqrt{ d\log\left(\frac{18L_2}{\sigma_2}\right)+\frac{1}{2}d\log(\frac{s}{d})+\log\left(\frac{6\sigma_2^2r\sqrt{s}}{\gamma_2\sigma_1{\delta}}\right)}$,\\
$L_2=\max\{L,L_1\}$. 
\noindent When $|1-\rho|<1$, \\$\limsup_{R\rightarrow \infty}\norm{\theta^R-\theta^*}\leq \frac{\alpha(2+\epsilon_3)(4\Gamma_1+\beta)}{\rho}$.
\end{theorem*}
%\vspace{-1.5mm}
\begin{proof}
From Section \ref{sec:diversefl}, we recall that for round $i\in [R]$, $\wtilde{\mathcal{N}}^{(i)}$ is the set of indices of the clients for which the per-client similarity conditions in \ref{eq:cond1} and \ref{eq:cond2} are satisfied. Therefore, for $j\in\wtilde{\mathcal{N}}^{(i)}$, the client's uploaded update $z_j^{(i)}$ and its corresponding guiding update $\wtilde{\Delta}_j^{(i)}$ update satisfy $z_j^{(i)}\cdot\wtilde{\Delta}_j^{(i)}>0$, and $\norm{z_j^{(i)}}/\norm{\wtilde{\Delta}_j^{(i)}}<\epsilon_3$. As $E=1$, $\wtilde{\Delta}_j^{(i)}=\alpha^i\wtilde{g}_j^{i,1}$. For simplicity, we define $\wtilde{g}_j^{i}\overset{\Delta}{=}\wtilde{g}_j^{i,1}$. For proving our Theorem, we first need multiple Lemmas that are described next.
%\vspace{-1.5mm}
\begin{lemma}
\label{lemma:1}
For $i\in[R]$ and $j{\in}\wtilde{\mathcal{N}}^{(i)}$, we have the following:
\begin{align}
\label{eq:lemma1}
    \norm{z_j^{(i)}{-}\alpha^i(\grad F_j(&\theta^{(i-1)}){-}\grad F_j(\theta^*))}\nonumber\\&{\leq} (2{+}\epsilon_3)\alpha^i\norm{\wtilde{g}_j^{i}{-}\grad F_j(\theta^{(i-1)})}\nonumber\\
    &+(1{+}\epsilon_3)\alpha^i\norm{\grad F_j(\theta^{(i-1)}){-}\grad F_j(\theta^*)}\nonumber\\
    &+(2+\epsilon_3)\alpha^i \beta
\end{align}
\end{lemma}
%\vspace{-1.5mm}
\begin{proof}
\begin{align}
&\norm{z_j^{(i)}-\alpha^i(\grad F_j(\theta^{(i-1)})-\grad F_j(\theta^*))}\nonumber\\
&\leq \norm{z_j^{(i)}-\wtilde{\Delta}^i_j}+\norm{\wtilde{\Delta}^i_j-\alpha^i\grad F_j(\theta^{(i-1)})}+\alpha^i\norm{\grad F_j(\theta^*)},\nonumber\\
&\overset{a}{\leq}\norm{z_j^{(i)}+\wtilde{\Delta}^i_j}+\norm{\wtilde{\Delta}^i_j-\alpha^i\grad F_j(\theta^{(i-1)})}+\alpha^i\norm{\grad F_j(\theta^*)},\nonumber\\
&\overset{b}{\leq}(1+\epsilon_3)\norm{\wtilde{\Delta}^i_j}+\norm{\wtilde{\Delta}^i_j-\alpha^i\grad F_j(\theta^{(i-1)})}+\alpha^i\norm{\grad F_j(\theta^*)},\nonumber\\
&\quad+\norm{\wtilde{\Delta}^i_j-\alpha^i\grad F_j(\theta^{(i-1)})}+\alpha^i\norm{\grad F_j(\theta^*)},\nonumber\\
&\overset{(c)}{\leq}(2+\epsilon_3)\alpha^i\norm{\wtilde{g}_j^{i}-\grad F_j(\theta^{(i-1)})}\nonumber\\
&\quad+(1+\epsilon_3)\alpha^i\norm{\grad F_j(\theta^{(i-1)})-\grad F_j(\theta^*)}\nonumber\\
&\quad+(2+\epsilon_3)\alpha^i \beta,
\end{align}
where $(a)$ holds as $z_j^{(i)}{\cdot}\wtilde{\Delta}_j^{(i)}{>}0$, while $(b)$ holds because $\norm{z_j^{(i)}}{/}\norm{\wtilde{\Delta}_j^{(i)}}{<}\epsilon_3$. For $(c)$, note that $\theta^*$ satisfies the following:
\begin{equation}\frac{1}{N}\sum_{j\in[N]}\grad F_j(\theta^*)=0.\end{equation} Therefore, by Assumption $3$, $\norm{\grad F_j(\theta^*)}\leq \beta$.
\end{proof}
%\vspace{-1.5mm}
\begin{lemma}
\label{lemma:2}
Let the learning rate be $\alpha^i=\alpha=\mu/(2L^2)$ for each communication round $i\in[R]$. Then, the following holds:
\begin{align}
\norm{\theta^{(i-1)}{-}\theta^*{-}\alpha^i(\grad F_j(\theta^{(i-1)}){-}\grad F_j(\theta^*))}\nonumber\\
{\leq} \sqrt{1{-}\frac{\mu^2}{4 L^2}}\norm{\theta^{(i-1)}{-}\theta^*}\nonumber
\end{align}
\end{lemma}
\begin{proof} Expanding the left hand side, we have:
\begin{align}
&\norm{\theta^{(i-1)}-\theta^*-\alpha^i(\grad F_j(\theta^{(i-1)})-\grad F_j(\theta^*))}^2\nonumber\\
&\quad=\norm{\theta^{(i-1)}-\theta^*}^2+\alpha^2\norm{\grad F_j(\theta^{(i-1)})-\grad F_j(\theta^*)}^2\nonumber\\
&\quad\quad-2\alpha(\theta^{(i-1)}-\theta^*)\cdot(\grad F_j(\theta^{(i-1)})-\grad F_j(\theta^*)).\label{eq:e}
\end{align}
By Assumption $1$, we have the following:
\begin{align}
    \norm{\grad F_j(\theta^{(i-1)})-\grad F_j(\theta^*)}\leq L\norm{\theta^{(i-1)}-\theta^*},\label{eq:a}
\end{align}
\begin{align}
    F_j(\theta^*){+}\grad F_j(\theta^*){\cdot}(\theta^{(i-1)}{-}\theta^*)&{\leq} F_j(\theta^{(i-1)})\nonumber\\
    &{-}\frac{\mu}{2}\norm{\theta^{(i-1)}{-}\theta^*}^2,\label{eq:b}
\end{align}
\begin{align}
    &F_j(\theta^{(i-1)})+\grad F_j(\theta^{(i-1)})\cdot(\theta^*-\theta^{(i-1)})\leq F_j(\theta^*).\label{eq:c}
\end{align}
Summing up \eqref{eq:b} and \eqref{eq:c} results in the following:
%\vspace{-1.0mm}
\begin{align}
    (\theta^*{-}\theta^{(i-1)})&{\cdot} (\grad F_j(\theta^{(i-1)}){-}\grad F_j(\theta^*))\nonumber\\
    &{\leq} {-}\frac{\mu}{2}\norm{\theta^{(i-1)}{-}\theta^*}^2.\label{eq:d}
\end{align}
%\vspace{-1.0mm}
Substituting \eqref{eq:a} and \eqref{eq:d} in \eqref{eq:e}, we have the following:
\begin{align}
&\norm{\theta^{(i-1)}-\theta^*-\alpha(\grad F_j(\theta^{(i-1)})-\grad F_j(\theta^*))}^2\nonumber\\
&\quad\quad\quad\quad\quad\quad\leq (1+\alpha^2 L^2-\alpha^i\mu)\norm{\theta^{(i-1)}-\theta^*}^2.\label{eq:f}
\end{align}
Taking square root and using  $\alpha{=}\frac{\mu}{2L^2}$ in \eqref{eq:f} completes the proof.
\end{proof}
% Lemma 3
% \vspace{-1.5mm}
The following Lemma is adapted from Lemma 4 in \cite{cao2020fltrust} and reproduced here for completeness.
% \vspace{-1.5mm}
\begin{lemma}
\label{lemma:3}
For any $\delta{\in}(0,1)$, we define the following:
\begin{align}
\Gamma_1&=\sqrt{2}\sigma_1\sqrt{(d\,{\log 6}+\log({3{/}\delta})){/}s},\nonumber\\
\Gamma_2&=\sigma_2\sqrt{\frac{2}{s}}\sqrt{ d\log(\frac{18L_2}{\sigma_2}){+}\frac{1}{2}d\log(\frac{s}{d}){+}\log\left(\frac{6\sigma_2^2r\sqrt{s}}{\gamma_2\sigma_1{\delta}}\right)},\nonumber\\ 
L_2&=\max\{L,L_1\}. \nonumber
\end{align}
When $\Gamma_1\leq \sigma_1^2/\gamma_1$, and $\Gamma_2\leq \sigma_2^2/\gamma_2$, we have the following for each client $j\in\wtilde{\mathcal{N}}^{(i)}$ in communication round $i\in[R]$:
\begin{align}
    \Prob&\left(\norm{\wtilde{g}_j^i-\grad F_j(\theta^{(i-1)})}\leq8\Gamma_2\norm{\theta^{(i-1)}-\theta^*}+4\Gamma_1\right)\nonumber\\
    &\geq1-\delta.
\end{align}
\end{lemma}
Next, we leverage the above results to prove our Theorem. For communication round $i\in[R]$, we have the following:
\begin{align}
&\norm{\theta^i-\theta^*}
=\Biggl|\!\!\!\Biggl|\,\theta^{(i-1)} - \frac{\alpha}{\ssn}\sum_{j\in\ssn}(\grad F_j(\theta^{(i-1)})-\grad F_j(\theta^*))\nonumber\\
    &+\frac{\alpha}{|\ssn|}\sum_{j\in\ssn}(\grad F_j(\theta^{(i-1)})-\grad F_j(\theta^*))-  \frac{1}{|\ssn|}\sum_{j\in\ssn}z_j^i-\theta^*\,\Biggr|\!\!\!\Biggr|\,\,,\nonumber\\
&\overset{(a)}{\leq}\underbrace{\frac{1}{|\ssn|}\sum_{j\in\ssn}\norm{\theta^{(i-1)}-\alpha(\grad F_j(\theta^{(i-1)})-\grad F_j(\theta^*))-\theta^*}}_{e_1}\nonumber\\
    &+\underbrace{\frac{(2+\epsilon_3)\alpha}{|\ssn|}\sum_{j\in\ssn}\norm{\tilde{g}_j^i-\grad F_j(\theta^{(i-1)})}}_{e_2}\nonumber\\
    &+\underbrace{\frac{(1+\epsilon_3)\alpha}{|\ssn|}\sum_{j\in\ssn}\norm{\grad F_j(\theta^{(i-1)})-\grad F_j(\theta^*)}}_{e_3}\nonumber\\
    &+\underbrace{{(2+\epsilon_3)\alpha}\grad F_j(\theta^*)}_{e_4},\nonumber
\end{align}
where $(a)$ follows from Lemma \ref{lemma:1}. For bounding the above, we note that $e_1$, $e_3$ and $e_4$ are bounded based on Lemma \ref{lemma:2}, Assumption $1$ and Assumption $3$ respectively. Furthermore, by Lemma \ref{lemma:3}, each term in the summation in $e_2$ is bounded by $(8\Gamma_2\norm{\theta^{(i-1)}-\theta^*}+4\Gamma_1)$ with probability at least $(1-\delta)$. Hence, by Fréchet lower bound for the probability of intersection of events, $e2$ can be bounded by $(8\Gamma_2\norm{\theta^{(i-1)}-\theta^*}+4\Gamma_1)$ with probability at least $1-|\ssn|\delta \geq 1-N\delta$. Hence, by defining $\wtilde{\delta}\triangleq N\delta$, we have the following with probability at least $1-\wtilde{\delta}$:
\begin{align}
    &\norm{\theta^i-\theta^*}\nonumber\\
    &{\leq}\left(\sqrt{1-\frac{\mu^2}{4L^2}}+8\alpha(2+\epsilon_3)\Gamma_2+\alpha(1+\epsilon_3)L\right)\norm{\theta^{(i-1)}-\theta^*}\nonumber\\
    &\quad\quad\quad\quad+4\alpha(2+\epsilon_3)\Gamma_1+\alpha(2+\epsilon_3)\beta.\label{eq:i}
    \end{align}
Recursively applying \eqref{eq:i}, we arrive at \eqref{eq:thm} in the Theorem.
\end{proof}

%\newpage
\iffalse
\input{12-Appendix-DataCleaning}
\fi

% \ifCLASSOPTIONcaptionsoff
%   \newpage
% \fi
% \appendices
% \input{7-appendices}

\end{document}